\documentclass[10pt,final,journal,twocolumn]{IEEEtran}
\usepackage{graphicx}
\usepackage{amsmath}
\usepackage{amsthm}
\usepackage{amsfonts}
\usepackage{amssymb}
\usepackage{comment}
\usepackage[numbers]{natbib}
\usepackage{hyperref}
\usepackage[utf8x]{inputenc}
\usepackage{lipsum}
\usepackage{tikz}
\usetikzlibrary{trees}
\tikzstyle{level 1}=[level distance=3cm, sibling distance=3cm]
\tikzstyle{level 2}=[level distance=3cm, sibling distance=2cm]

\tikzstyle{bag} = [text width=4em, text centered]
\tikzstyle{end} = [circle, minimum width=3pt,fill, inner sep=0pt]

\usepackage{mathtools, stmaryrd}

\newtheorem{prop}{Proposition}
\newtheorem{coro}{Corollary}
\newtheorem{defi}{Definition}

\title{ABBA: A quasi-deterministic Intrusion Detection System for the Internet of Things}
\author{
	\IEEEauthorblockN{Dr. Raoul Guiazon}\\
	\IEEEauthorblockA{Faculty of Engineering and Physical Sciences,\\ University of Leeds,\\ Leeds LS2 9JT, United Kingdom.} 
	\thanks{Thanks to the Royal Academy of Engineering and the Office of the Chief Science Adviser for National Security for supporting this work under the UK Intelligence Community Postdoctoral Fellowship Programme.}\thanks{Thank you also to Dr. Daphne Tuncer for her help in writing this article.}
	\thanks{\textbf{Contact: guiazonraoul@gmail.com}}}
\begin{document}
\maketitle

\begin{abstract}
An increasing amount of processes are becoming automated for increased efficiency and safety. Common examples are in automotive, industrial control systems or healthcare. Automation usually relies on a network of sensors to provide key data to control systems. One potential risk to these automated processes comes from fraudulent data injected in the network by malicious actors.  In this article we propose a new mechanism of data tampering detection that does not depend on secret cryptographic keys - that can be lost or stolen - or accurate modelling of the network as is the case with existing machine learning based techniques. We define and analyse the mathematical structure of the proposed technique called ABBA and propose an algorithm for implementation.	
\end{abstract}
\section{Introduction}

A vast amount of our modern infrastructure relies on a network of sensors and actuators to automatically perform various tasks. The purpose of this automation is generally to increase efficiency and decrease waste and human risk. This has typically been the case for industrial control systems (ICS) that often form the core of critical national infrastructures (CNI). In recent years the topology of the networks underpinning these infrastructures has changed, moving from somewhat isolated networks using dedicated and often proprietary technologies to transport and process data to more standard internet connected devices that often rely on cloud processing to make critical decisions (figure \ref{Net_top}).\\
Although this new topology increases the capabilities and flexibility of these networks by allowing more interoperability between systems and a better understanding and monitoring of assets and processes, it also opens up more vulnerabilities for cyber criminals to enter networks.\\
Cyber attacks on CNIs are not theoretical, in 2015 the world witnessed the first known power outage caused by a malicious cyber attack that happened when utility companies in Ukraine were hit by the BlackEnergy malware. In February 2021 a water treatment facility in Florida was attacked, the attacker remotely increased the levels of sodium hydroxide content from 100 parts per million to 11,100ppm putting at risk 15000 people relying on this plant for clean water. More recently on May 7 2021, the US issued emergency legislation after Colonial Pipeline which carries almost half of the East Coast supply of diesel petrol and jet fuel was hit by a ransomware cyber-attack.\\
In the Siemens report "Caught in the Crosshairs: Are utilities Keeping Up with the Industrial Cyber Threat?" it is found that 30\% of attacks on OT (Operational Technologies) are not detected. Given the complexity of the systems utilised to automate our infrastructures, cars and cities with often millions of lines of codes running on top of various hardware, it is impossible to guarantee that no vulnerabilities will ever be found and be exploited by adversaries.\\
Cyber criminals exploit multiple routes into their target systems, from phishing attacks to software vulnerabilities. Bad practices from a supplier also can have important repercussions further down the chain. For example, hard-coded admin credentials on a device or a key server breach at a device manufacturer can enable an attack on the end user's network.  The weakest link will be the point of entry into the network.\\
The purpose of this work is to help secure the fleet of small devices that often relay sensing data to a network controller to be processed an relied upon for critical decision making, this could be for the Advance driver-assistance system (ADAS) of a vehicle or an ICS.  We consider that the aim of the attacker is to modify the behaviour of the automated system under attack by feeding fraudulent data to the decision making unit. In essence, we are interested in developing a method to detect such an attack even in the case where the attacker has a copy of the secret key used by a legitimate device to authenticate with the network. 

\begin{figure}
	\centering
	\includegraphics[scale=0.15]{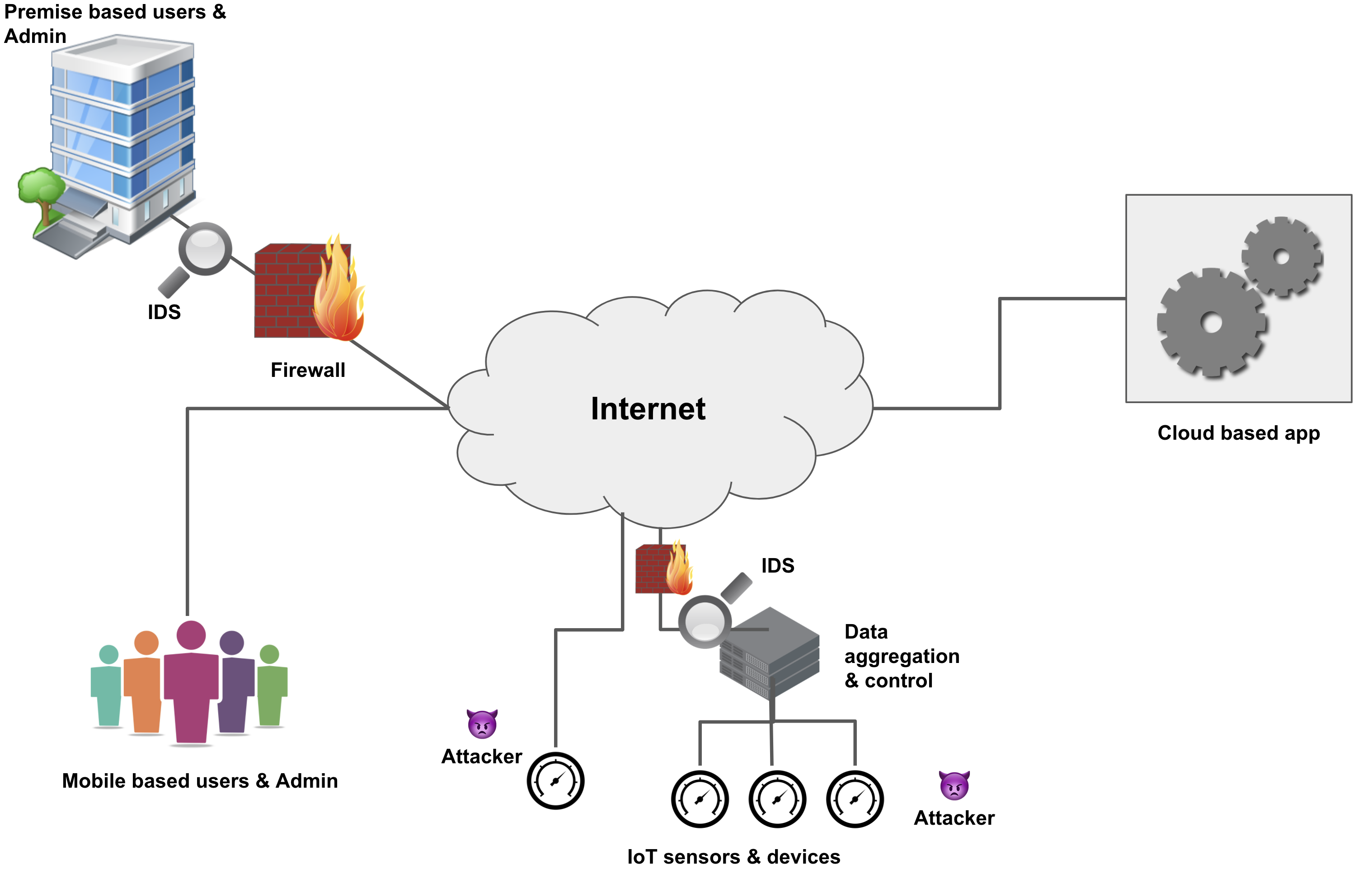}
	\label{Net_top}
	\caption{Modern network topology}
\end{figure}

Multiple approaches are often utilised to protect digital networks from cyber threats and all participate in making those infrastructures safer. The most common layer of protection is often based on cryptographic techniques to ensure confidentiality, integrity or authenticity of the traffic \cite{7172449}. This layer requires the generation and distribution of secret keys and the management of these keys over the life of a device. Often enough this first layer is inexistent, poorly implemented \cite{fi12030055} or  become obsolete due to new vulnerabilities found in key protocols \cite{5601275}.  Other methods focus on the physical layer of the communication stack using channel state estimation to generate secret keys \cite{9270035} or using jamming to reduce the signal quality of an eavesdropper \cite{8627099,7524448}. These methods are not adapted for situations where channel variations are limited or where jamming could affect neighbouring networks.\\
To complement the techniques mentioned before, some systems also implement intrusion detection systems (IDS). IDS often function at a higher level of abstraction, monitoring access permissions and traffic patterns which they compare to a baseline behaviour considered normal for a specific network.
Current IDS techniques are often built using machine learning techniques or rules based methods \cite{8365277}. With machine learning, the IDS needs reliable training data to learn the “normal patterns” on the sensor network, in this case the problem is that defining what “normal” means is a challenge. Any bias in the training data will increase the risk of false positives and false negatives. With rule-based methods, attack signatures are listed in the IDS memory to enable it to identify similar attacks in the future. This method fails against zero-day attacks and requires the IDS to be kept up to date when new attacks are discovered. Moreover, IDS systems are often checking the traffic from the internet into the internal network and less so the traffic coming from the sensors.\\
The technique devised in this paper does not rely on training a model, maintaining an attack signature list or managing secure keys although these can be complementary. We call this method Artificial Behaviour Based Authentication (ABBA), it is a mechanism by which a device or multiple devices generate a pattern in their network to facilitate the detection of anomalies and intruders on the network. The pattern created is the artificial behaviour of the network which is built using a time code defined in this document.​\\

In the following sections we build the theoretical framework for a robust Intrusion Detection System (IDS) that monitors the communication link from the sensors to the core network for signs of cyber attacks and enables early detection of cyber incidents. We put in place the framework to assess the detection probability of such IDS and provide an algorithm to implement this technique on a physical system. In this work we purposely do not dive into the exact implementation of the technique as this depends on multiple parameters that can be optimised for each individual application. One such implementation will be described in our github repository \url{https://github.com/abbaiot}. Here we focus on describing ABBA and demonstrating its ability to reliably detect anomalies such as data losses or data injections and tampering due to a third party.\\
The structure of this paper is as follows, section \ref{Sys_mod} lays out the structure of ABBA and the different key elements that are required to make it work. Section \ref{Time_encoder} describes the theory behind the time encoder that generates the signatures used to authenticate a device. Section \ref{encoding} introduces the encoding used for ABBA and describes its properties. In section \ref{IDS} we put together all the pieces of the IDS and describe its behaviour under different types of attacks. Section \ref{Tea} presents the practical algorithm behind ABBA and finally, section \ref{Clc} concludes this paper.\\

{\em Notations}\\
$\mathbb{N}$ is the set of natural numbers.\\
$\mathbb{R}^+$ is the set of positive real numbers.\\
$\{\cdot\}$ represents a set.\\
$(h)$ represents a sequence with elements $h_n$ with $n\in \mathbb{N}$.\\
$(x,y)$ is a pair.\\
$((x,y))$ is a sequence of pairs.

\section{System model}\label{Sys_mod}
We consider an information source $S$ that produces events from a finite alphabet $X=\{x_1,...,x_N\}$ according to the probability distribution $P_X$ and two parties Alice and Bob where Alice observes the events produced by the information source and communicates that information to Bob over an untrusted communication channel. As is conventional in security research we also refer Eve as the eavesdropper or attacker on the network.\\
We are not concerned with distortions or losses during transmissions from Alice to Bob therefore we consider the communication channel perfectly error corrected. Only the actions of Eve are of concern to us.\\
We devise a mechanism with either of the following two desired properties,
\begin{enumerate}
	\item Bob can detect malicious actions by Eve with a computable detection probability. 
	\item Alice and Bob can trap Eve in the network for a given duration during which she has to spend computing resources to remain undetected whilst increasing chances of her being detected using other techniques.
\end{enumerate}
We want that detection probability and duration to only depend on the specific parameters chosen by Alice and Bob as they implement the intrusion detection scheme.\\
\begin{figure}
	\centering
	\includegraphics[scale=0.15]{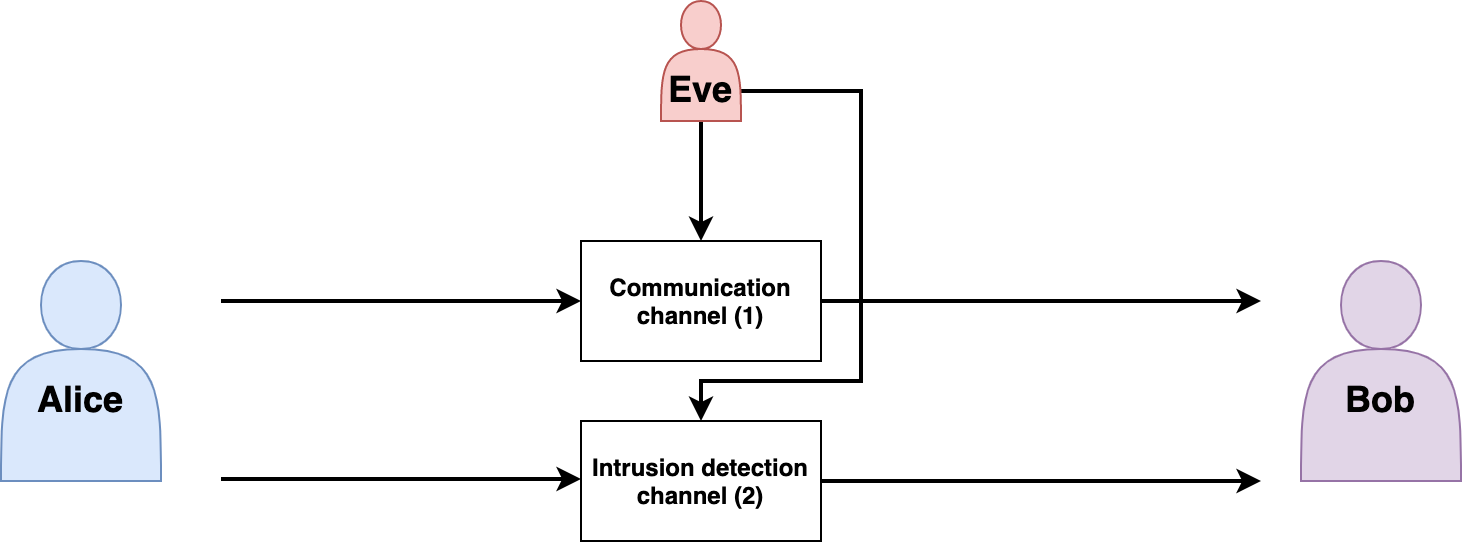}
	\label{Channels}
	\caption{Top level principle}
\end{figure}
The basic principle of the IDS proposed is described on figure \ref{Channels}, where Alice sends messages to Bob over 2 channels. Channel 1 is the "high bandwidth" communication channel that carries the payload and Channel 2 is the "Low bandwidth" intrusion detection channel (IDC). The reason for this architecture is to decouple the message from Alice to Bob to the signatures that allow Bob to verify the source of the message. This could be understood as - although not completely accurate - sending a message authentication code for a payload into a different channel than the payload. The aim would be for Bob to detect Eve's actions by matching the signature received on the IDC with the payload received.\\
The model described above cannot in itself prevent an attacker to infiltrate the network undetected. As they can tamper with both channels in a way that remains consistent to the receiver.\\
To solve this issue, we design an IDS that relies on only two assumptions for security that we argue are simple to guarantee or verify in any specific application. 
\begin{enumerate}
	\item Eve cannot prevent messages on the IDC from reaching Bob.
	\item Eve cannot delay or speed-up messages transiting on the IDC.
\end{enumerate}
For example, assumption 1 can be easily checked in a wireless link by monitoring the noise levels for jamming in the intrusion detection channel between Alice and Bob. A higher level of noise would decrease confidence in the IDS in a way reminiscent of what is used in Quantum key distribution systems. Assumption 2 depends on the topology of the network, for example, a wireless link between Alice and Bob doesn't allow for Eve to manipulate signals flight time. This is also applicable to some wired networks if the threat modelling discounts a Man-in-The-Middle type attack.

\begin{figure*}
	\centering
	\includegraphics[scale=0.07]{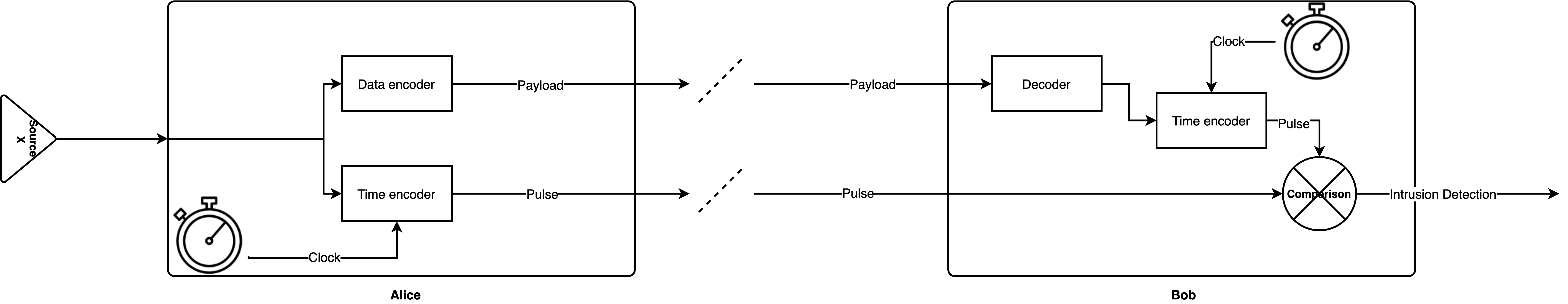}
	\label{Sys_model}
	\caption{Internal architecture of the IDS.}
\end{figure*}

The figure \ref{Sys_model} illustrates the internal structure of the system at both Alice and Bob. transmitter and receiver both rely on a local clock to implement the intrusion detection algorithm. We will assume that both clocks measure time at the same rate - compensating for local drifts would be required in practice but is not the focus of this paper. The key component of interest to us is the "Time encoder" that we will design in the following sections of this article. The time encoder produces pulses that are sent to Bob using a modulation as simple as ON/OFF keying (OOK) \cite{6612714,7399931} requiring low bandwidth and accessible to all devices. 

\section{Time encoder}\label{Time_encoder}
To describe how the time encoder used in the IDS works, we need an appropriate formalism which we develop in this section. We will start by defining how we represent information sources, then describe how one could devise a way to communicate the entire information content of an information source using an encoding solely based on the passage of time. We will then use our new formalism to define the time encoder that is core to the intrusion detection system proposed.
\subsection{Complete description of Information Sources}
In information theory, an information source $S$ is usually described using a random variable $X$ with values in a given set of the same name $X$. In this work we will limit ourselves to discrete sets, meaning that there exists a one to one mapping between $X$ and $\mathbb{N}$ the set of natural numbers.\\
In a typical experiment with an information source, a sequence of events from the set $X$ is produced, for example $(x_1,...,x_p)$ where the index of each element represents the position of the event in the sequence. Each event $x_i$ could be the outcome of a coin toss.\\
Because all our experiments are done with time evolving in the background, that sequence can always be implicitly redefined as $((t_1,x_1),...,(t_p,x_p))$ where $t_i$ are the time instants at which these events occurred with the origin of time that can be set arbitrarily before the experiment. Often this information about time is not useful, for example when writing a document, the time at which a key was pressed doesn't matter only the order and value of the keys do. To provide a complete physical description of the source it is however important to include time.\\
Once observed, the sequence $(t)=(t_1,...,t_p)$ is an increasing function from $\mathbb{N}$ to $\mathbb{R}^+$ the positive real number line and the sequence $(x) = (x_1,...,x_p)$ is a function from  $\mathbb{N}$ to $X$ with no particular properties besides that the frequencies of elements in the sequence will converge to their probabilities as defined by the random variable $X$ when $p$ increases.\\
The sequence $((t,x))$ is thus a function from $\mathbb{N}$ to $\mathbb{R}^+ \times X$.\\

If we were able to observe the entire sequence of events produced by our information source, we would build a potentially infinite sequence $(t,x)=((t_1,x_1),...,(t_p,x_p),...)$ at which point the information source would have been described entirely. In this document we define the information source to be that sequence $S=((t,x))$.

\subsection{A time-based encoding} 

The foundation of modern digital communications is the encoding of information as binary digits. Hence, before information from our source $S$ is transmitted, it is translated into a sequence of bits by a data encoder that are then used to modulate a carrier signal and sent onto the communication channel.\\
Our objective is to build an intrusion detection mechanism that requires minimum bandwidth utilisation. To achieve this goal we define a new encoding based on the passage of time at both transmitter and receiver. The code generated using that encoding can be used to modulate a carrier using ON-OFF keying to transmit data to Bob. \\
The main idea behind our time encoding is to define characters as time intervals of a given duration starting from a fixed origin in time. As an analogy, let's say we could represent time intervals as pieces of strings of various lengths, we could then use a substitution mechanism as described in figure \ref{String_code} to encode text written in english.
However, this simple substitution mechanism wouldn't work with time intervals if a single origin of time was defined for all letters because repetitions such as with the letter "l" in the word "Hello" would not be possible to represent and information about the order of these letters would be lost too. 
In the example shown in figure \ref{String_code} we represent every letter independently using a string with variable length for every letter of the word "Hello". Using a single dimension and marking the different lengths on the one string, the same word would look like figure \ref{Flat_string} where information about the order of the strings and their count is lost.
\begin{figure}
	\centering
	\includegraphics[scale=0.15]{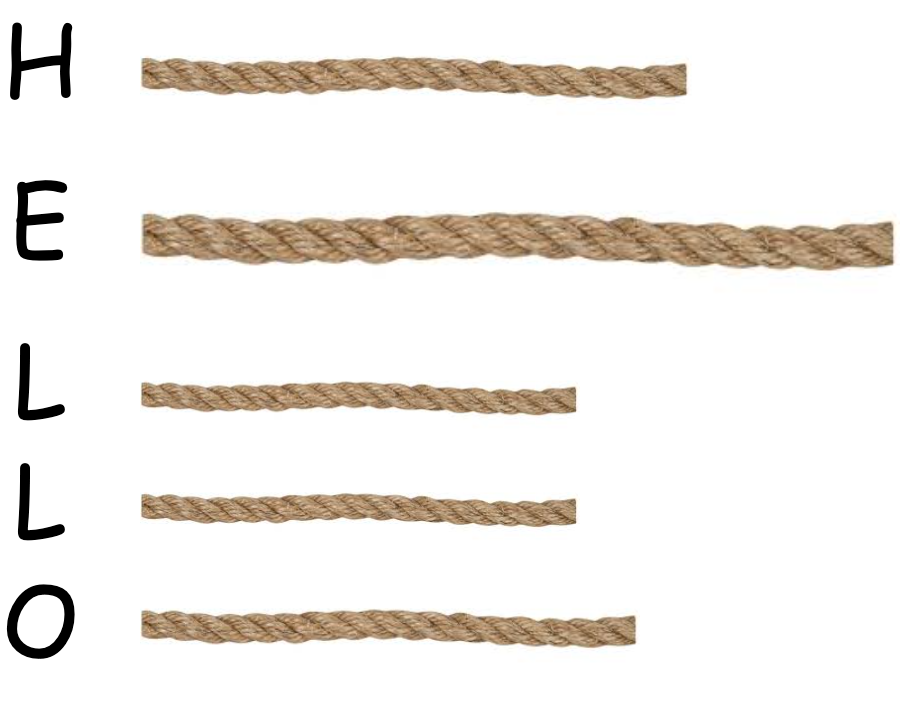}
	\caption{This figure shows how the word "Hello" can be written in pieces of strings of various length}\label{String_code}
\end{figure}
\begin{figure}
	\centering
	\includegraphics[scale=0.15]{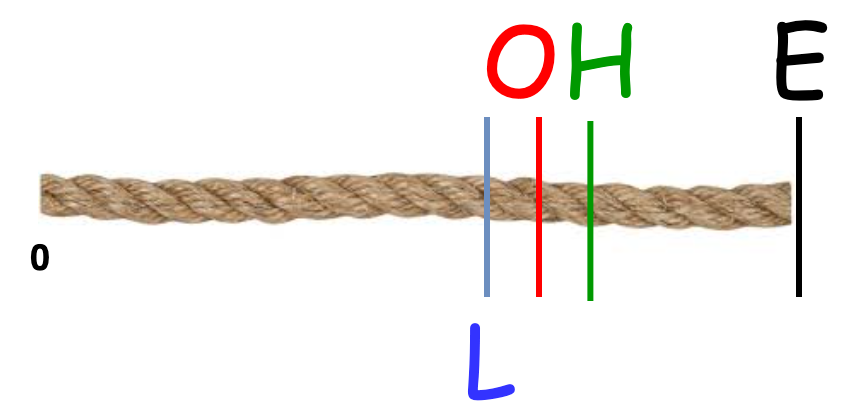}
	\caption{Projection of the different words onto a single dimension}\label{Flat_string}
\end{figure}

In the next section we show the existence of a time-code that preserves all information about the message transmitted. 

\subsection{Information preserving time-code} 
Before we take a look at our information source, let's first work out how we could define a natural time encoding of any positive real number as a time interval.\\
If we denote $\mathcal{T}$ the set of all time intervals in time units. The encoding of any positive real number is given by the one to one mapping $\phi: \mathbb{R}^+ \rightarrow \mathcal{T}$. We will denote elements of $\mathcal{T}$ with the greek letter $\tau$.\\
This means that for example, the number 10.234 is represented by a segment of time $\tau_{10.234}$ of a fixed duration in time units (whatever unit is chosen).\\
With this encoding one can represent numbers using chunks of time, for example, numbers from 1 to 5 are $\tau_1, \tau_2,\tau_3,\tau_4,\tau_5$. \\
Now let's go back to the complete description of our information source by the sequence $(t,x)=((t_1,x_1),...,(t_p,x_p),...)$, and discuss why this sequence can be mapped to a unique sequence of time intervals $(\tau)= (\tau_1,...,\tau_p,...)$.\\
The space $\mathbb{R}^+\times \mathbb{N}$ can be mapped one to one with $\mathbb{R}^+$ and we have shown a one to one mapping $\phi$ between $\mathbb{R}^+$ and $\mathcal{T}$.\\
For example the following map $f : \mathbb{R}^+\times \mathbb{N} \rightarrow  \mathbb{R}^+$ defined by 
\begin{equation}
f(t,n) = n+ \frac{1}{1 + \exp^t}
\end{equation}
With the inverse being
\begin{equation}
f^{-1}(y) = \left(\ln\left(\frac{1}{y-floor(y)}-1\right),floor(y)\right) 
\end{equation}
We denote $\Theta$ the one to one mapping $\Theta: (\mathbb{R}^+ \times X) \rightarrow  \mathcal{T}$. In which case
\begin{equation}
(\tau)= \Theta\left( (t,x) \right)
\end{equation}
Now we realise that the sequence $((t,x))=((t_1,x_1),...,(t_p,x_p),...)$ doesn't need to be written as an ordered list as the time information is already contained within each element. Instead, it can be represented as a set $S = \{(t_1,x_1),...,(t_p,x_p),...\}$ containing every element. Similarly, the time sequence $(\tau)= (\tau_1,...,\tau_p,...)$ can be represented by the set $T = \{ \tau_1,...,\tau_p,... \}$. 
The set $T$ is a codeword of time intervals that uniquely describes the information source $S$.\\

Note that we can define a superset $\mathcal{S}$ of sets like $S$ where each element represents a possible sequence of events with elements in $X$.  $\mathcal{S}$ is a source of information sources and generates all possible sequences of events. $\Theta$ can then be used to map each element of $\mathcal{S}$ to elements of the superset $\Psi$ of time sequences.\\
The mapping $\Theta$ is the dictionary used to describe elements of $\mathcal{S}$ in terms of codewords in $\Psi$.\\
We can now describe an information source as an element of $\mathcal{S}$ or equivalently as a codeword in $\Psi$.

\section{Communication using time-code}\label{encoding}
 
Communicating using the dictionary $\Theta$ described in previous sections is not practical because it requires Alice to know the entire history of the information source $S$ to devise a codeword $T$ that she can transmit to Bob. If she was able to do that then, a simple thing for Alice to do would be to use OOK modulation to transmit a pulse at the end of every time interval in the set $T$ - counting from a time $t_0$ when Alice and Bob synchronised their clocks. If Bob recorded the time of detection of pulses from time $t_0$ then he would be able to reconstruct $T$ after a possibly infinite amount of time. Then using $\Theta^{-1}$ he would be able to translate from $\Psi$ back to $\mathcal{S}$.\\ 
Because Alice and Bob can only move forward in time and usually have a limited lifetime, they need a coding scheme where the instants at which pulses are produced only depend on events that happened in their past and decoding can be done on the fly.\\
In the next section, we will build a causal time-code that Alice can use to convert her information source output into time intervals in a way that doesn't require her to know the future states of the source. The cost of that will be Bob's uncertainty about the information sent by Alice. It is because of that uncertainty that the IDS we propose is probabilistic in nature although, with a detection probability that approaches 1 in some cases as time increases. This will be discussed in section \ref{IDS}.
 
\subsection{A causal time-code}\label{causal_encoding}
 
The aim of this section is to build a practical encoding that both Alice and Bob can use to communicate enough information about the source to detect tampering using the intrusion detection channel.\\
Let's define a family of sequences $(g_s)$ (with values in $\mathbb{R}^+$ and $s \in \mathbb{R}^+$) and the family of functions $O_x : \mathbb{R}^+ \rightarrow \mathbb{R}^+$ with $x \in X$.\\

The encoding starts with two initial parameters, a number $s_0$ and time origin $t_0$.
Then the set $T$ is created based on  $((t,x))=((t_1,x_1),...,(t_p,x_p),...)$ as follows.
\begin{enumerate}
	\item Define the sequence $(s)$ with first term $s_0$ and $\forall n>0, s_{n} = O_{x_n}(s_{n-1})$.
	\item Define the sequences $(\tau_k)$, s.t $\tau_{k,0}=t_k + g_{s_k,0}$ and $\tau_{k,n+1} = \tau_{k,n}+ g_{s_k,n}, \forall k,n\geq 0$.
	\item  Now the set $T= \{\tau_{k,n}; \tau_{k,n}<t_{k+1}, \forall k,n\}$ is the codeword describing the source $S$ represented by $((t,x))$.
\end{enumerate}

It is easy to check that with this encoding, any two different codewords $T_1$ and $T_2$ would necessarily come from different sources $S_1$ and $S_2$. This property is critical for our intrusion detection system.\\
It is also true that two sources $S_1$ and $S_2$ could be mapped to the same output $T$. However, the parameters used in the encoding can be tuned to reduce the likelihood of such collisions and mitigate the risk of this being used by a potential attacker. The properties needed for these parameters are defined in appendix.\\
In summary the properties of this encoding are:
\begin{enumerate}
	\item Unlike any general encoding over the space $\Psi$, events produced by a source at time $t$ do not influence any letter $\tau$ of the corresponding codeword such that $\tau < t$.
	This means that this encoding can be generated by Alice and Bob.
	\item Neither Alice, Bob nor Eve can predict the next letters of the codeword sequence $(\tau)$ that encodes the source $S$ within a big enough time window if the entropy of the source is non zero.
	\item Transmitting or receiving the sequence $(\tau)$ only requires a simple physical device that can measure time and generate or measure a pulse. 
	\item The pulses produced by this encoding will never collide in time with an event generated by the original source.
	\item This encoding is not one to one in both directions. Two different sources can be translated into the same codeword.
	
\end{enumerate}

\section{Intrusion detection system}\label{IDS}

In this section we put together the time encoding that we have developed in the previous section together with the model described in the system model to define the protocol of our intrusion detection system. We also describe the different types of attacks that can be perpetrated by Eve and discuss how those can be detected by Bob using the IDS based on the proofs and analyses given in appendix \ref{append} and \ref{append2}.
\subsection{IDS protocol}

The intrusion detection mechanism that Alice and Bob utilise works as follow:
\begin{enumerate}
	\item Alice and Bob agree on an initial $s_0$ and $t_0$. In the general case, these do not need to be kept secret.
	\item Alice sends messages $x \in X$ to Bob using the communication channel and the corresponding $\tau \in T$ through the intrusion detection channel.
	\item Bob receives $\hat{x}$ and $\hat{\tau}$ assumed to be from Alice. He then computes $\tilde{\tau}$ and compares it to $\hat{\tau}$.
	\item If $\tilde{\tau} \neq \hat{\tau}$ then he knows there is an anomaly in the stream coming from Alice.
\end{enumerate}

In the following sections we look at different types of attacks that Eve can perform on the communication channel and discuss their implications.

\subsection{Message injection by Eve}  
In this section, we consider the case where Eve injects messages and pulses on both the communication channel and the IDC.\\
Corollary \ref{coro1} in appendix \ref{append} shows that with well chosen parameters for the encoder, if Eve performs this attack, she can only remain undetected if she continues to send messages  after the latest transmission of Alice for a time duration fixed depending solely on the encoding. This traps Eve in a position where she needs to correct her influence on the data stream by adjusting for future messages that Alice could transmit on both channels, otherwise Eve will be detected.\\
The consequence is that after her first data injection, Eve doesn't have full control over the messages she is sending to Bob to conceal her presence as she needs to account for past and future messages from Alice. Because Eve will have to send additional payload to conceal her presence in the network, the sequence of messages received by Bob could look abnormal based on some "normal" metric defined for his expected messages. For example if Bob was expecting a sentence, the perturbations of Eve could generate non-sensical sentences at Bob's end. Additionally in this case, the attacker is required to continuously spend energy to maintain their presence on the network hidden to the IDS. This is an additional constraint put on them compared to other IDS.

\subsection{Message tampering and deletion by Eve} 
Eve could attempt to change the value of a message $x$ on its way to Bob. In this case, Proposition \ref{prop_tampering} shows us that she will not be detected only if the new value $x'$ that she wants to transmit is so that given the current value of $s$ used by Alice and Bob, $O_x(s)=O_{x'}(s)$. If the information source of Alice has non zero entropy then Eve cannot know ahead of time which $x'$ to send to match the value $x$ sent by Alice. This means that the detection of Eve depends on the probability that she selects a compatible $x'$ at random. This probability is also influenced by the functions and sequences used to build the time-encoder and can be computed given the specific choice of parameters for the encoder.\\
Eve could also avoid detection if her tampering with the data is randomly covered by further data sent by Alice. We study this case in Appendix \ref{append2} and show that the parameters of the encoding can be chosen so that the probability of Eve avoiding detection becomes vanishingly small as time goes on. This means that the actions of Eve will be detected given enough time.\\
Another possible attack by Eve can be to remove a message sent by Alice from the communication channel. If this happens then corrolary 2 shows that this would also be detected if the functions $O_x$ are chosen not to have a fixed point. And similarly to the previous attack, further transmissions by Alice might delay the detection of Eve however with a probability that can be made to vanish with time (see appendix \ref{append2}).

\section{Practical implementation}\label{Tea}
To implement the IDS proposed in section \ref{IDS}, we need an algorithm to build the codewords $T$ out of the sequence of messages produced by the source. We create such an algorithm that we refer to as Time Encoder in the following section.

\subsection{Time Encoder Algorithm}
\begin{enumerate}
	\item Initialisation
	\begin{itemize}
		\item Random value $\rightarrow s$
		\item $0 \rightarrow count$
		\item $0 \rightarrow rnd\_time$
		\item Create an empty list $\tau$
		\item $False \rightarrow interrupt$
		\item Create a variable $x$ to store a source message
		\item Set a timer $current\_time$ counting up from 0
		\item start function "Generate interrupt"
		\item Go to next step
	\end{itemize}
	\item Generate random time
	\begin{itemize}
		\item $g_s(count) + current\_time \rightarrow rnd\_time$  \# Generating a candidate $\tau$ value.
		\item $count+1 \rightarrow count$
		\item Go to next step
	\end{itemize}
	\item Selecting $\tau$
	\begin{itemize}
		\item While ($interrupt == False$ AND $current\_time < rnd\_time$)\{Wait\}
		\item IF $interrupt == False$\\
		THEN \\
		$(\tau,rnd\_time) \rightarrow  \tau$  \# Adding $rnd\_time$ to the list $\tau$.\\
		Go to step 2\\
		ELSE\\
		Go to next step
	\end{itemize}
	\item Generating s (new seed).
	\begin{itemize}
		\item $ O_x(s) \rightarrow s$  \#Note that $O_x$ depends on the message $x$.
		\item $0 \rightarrow count$  \# Resetting $count$.
		\item $False \rightarrow interrupt$  \# Resetting interrupt flag.
		\item Go to step 2
	\end{itemize}

\end{enumerate}

Function : Generate interrupt\\
While (1)
\{
\begin{itemize}
	\item Listen to the Source $S$
	\item IF a new message $msg$ is detected\\
	THEN \\
	$True \rightarrow interrupt$  \# Changing interrupt flag.\\
	$msg \rightarrow x$  \# Storing the message.
\end{itemize}
\}

\subsection{Remarks}
The implementation described above is a generic algorithm that can be implemented choosing specific functions that are adapted to the device and the information source.\\
Moreover, different applications might have different requirements such as, the speed of detection of an intrusion or the energy consumption of the encoding that will affect the choice of the functions and parameters chosen for the encoding.

\section{Conclusion}\label{Clc}
In this work we have developed a new theoretical framework for an intrusion detection system that detects data tampering and data injection on a communication channel. This framework allows to build a detection mechanism that doesn't rely on prior knowledge of the network behaviour or knowledge of possible attack signatures. It is also advantageous compared to message authentication codes in the quality that it doesn't require the prior exchange or management of secret keys for the receiver to detect tampering with the received data. We also detail an algorithm to adapt the IDS defined in this paper to any device.\\

\appendices

\section{Characterisation of the time encoding}\label{append}

\begin{prop}
Let's consider 2 sequences of pairs $((t,x))$ and $((t',x'))$ of finite length $p_1$ and $p_2$ respectively.  
Let's denote $s \mbox{ and } s'$ the last terms of the respective sequences $(s) \mbox{ and } (s')$. \\
Assuming that $t_{p_1}\geq t_{p_2}$, we show that $(\tau) = (\tau')$  $\implies \exists P \in \mathbb{N} \mbox{ s.t } \forall n\geq 1, g_{s}(n) = g_{s'}(n+P)$ and $t_{p_1} = t'_{p_2} + \sum_{k=0}^{P+1}g_{s',k} - g_{s,0}$\\
	
\end{prop}

\begin{proof}
$(\tau) = (\tau') \mbox{ and } t_{p_1}\geq t_{p_2} \implies \exists M\in \mathbb{N}$ such that $\forall n\geq 0, \tau_{p_1,n} = \tau'_{p_2,n+M}$ which means
\begin{align*}
\tau_{p_1,0}-\tau'_{p_2,M} &= 0\\
\tau_{p_1,1}-\tau'_{p_2,M+1} &= 0\\
\cdots &= \cdots \\
\tau_{p_1,M}-\tau'_{p_2,2M} &= 0\\
\tau_{p_1,M+1}-\tau'_{p_2,2M+1} &= 0\\
\cdots &= \cdots 
\end{align*}
Which we rewrite as 
\begin{align*}
t_{p_1} + g_{s,0} - \tau'_{p_2,M}&=0\\
t_{p_1}+g_{s,0}+g_{s,1} - (\tau'_{p_2,M}+g_{s',M})&=0\\
t_{p_1}+g_{s,0}+g_{s,1}+g_{s,2} - (\tau'_{p_2,M}+g_{s',M}+g_{s',M+1})&=0\\
\cdots &= \cdots \\
t_{p_1} + \sum_{k=0}^{M-1}g_{s,k} - (\tau'_{p_2,M} + \sum_{k=M}^{2M-1}g_{s',k}) &=0\\
t_{p_1} + \sum_{k=0}^{M}g_{s,k} - (\tau'_{p_2,M} + \sum_{k=M}^{2M}g_{s',k})&=0\\
\cdots &= \cdots 
\end{align*}

From consecutive pairs of equations listed above, we get that 
\begin{align*}
g_{s,1}&=g_{s',M}\\
g_{s,2}&=g_{s',M+1}\\
\cdots &= \cdots \\
g_{s,M}&=g_{s',2M}\\
\cdots &= \cdots 
\end{align*}
 
This means $\exists P \in \mathbb{N} \mbox{ s.t } \forall n\geq 1, g_{s}(n) = g_{s'}(n+P)$   with $P=M-1$.\\
For the second part \\
$\tau'_{p_2,M} = t'_{p_2}+\sum_{k=0}^{M}g_{s',k}$ and $ \tau_{p_1,0}= t_{p_1} + g_{s,0}$\\
Knowing that 
$\tau_{p_1,0}=\tau'_{p_2,M}$ leads to\\
\begin{equation}
t_{p_1} = t'_{p_2} + \sum_{k=0}^{M}g_{s',k} - g_{s,0}
\end{equation}

\end{proof}

\begin{prop}
Let's consider 2 sequences $((t,x))$ and $((t',x'))$ of length $p_1$ and $p_2$ respectively.  
Let's assume that the last terms $s \mbox{ and } s'$ of the respective sequences $(s) \mbox{ and } (s')$ are equal  $s=s'$.\\
We show that if $t_{p_1} > t'_{p_2}$ then $(\tau) = (\tau')$  $\implies$ $(g_s)_{n>0}$ is a periodic sequence with period $P>0$ and $t_{p_1} = t'_{p_2} + \sum_{k=1}^{P+1}g_{s,k}$. \\
\end{prop}

\begin{proof}
The proof is similar to that of proposition 1, with the difference that the equivalent set of equation yields
\begin{align*}
g_{s,1}&=g_{s,M}\\
g_{s,2}&=g_{s,M+1}\\
\cdots &= \cdots \\
g_{s,M}&=g_{s,2M}\\
\cdots &= \cdots \\
\end{align*}
We also show below that $M>1$.\\
By construction,\\
$\tau_{p_1,0} = t_{p_1} + g_{s,0}$\\
$\tau_{p_1,1} = t_{p_1} + g_{s,0}+ g_{s,1}$\\
$\tau'_{p_2,0} = t'_{p_2} + g_{s,0}$\\
$\tau'_{p_2,1} = t'_{p_2} + g_{s,0} + g_{s,1}$\\
Because $t_{p_1} \neq t'_{p_2}$ and $g_{s,0},g_{s,1}\geq 0$  we know that $\tau_{p_1,0} \neq \tau'_{p_2,0}$ and $\tau_{p_1,0} \neq \tau'_{p_2,1}$\\
Similarly, $\tau'_{p_2,0} \neq \tau_{p_1,0}$ and $\tau'_{p_2,0} \neq \tau_{p_1,1}$. This means that the shift between those sequences is strictly greater than 1.\\
This completes the proof that $(g_s)_{n>0}$ is periodic with period $P=M-1>0$.\\
And finally, from proposition 1 we get the result
\begin{equation}
t_{p_1} = t'_{p_2} + \sum_{k=1}^{P+1}g_{s,k}
\end{equation}
\end{proof}

\begin{defi}[Non-maximally correlated]
	A family of infinite sequence $\{(g_s),s \in \mathbb{R}^+ \}$ is said to be non-maximally correlated if  $\forall (s_a,s_b) \mbox{ with } s_a \neq s_b, ~\nexists P\in \mathbb{N}, \mbox{ s.t } \forall n>0, g_{s_a,n}=g_{s_b,n+P}$. 
\end{defi}

\begin{coro}\label{coro1}
If the family of sequence {$(g_s)$} chosen for the encoding is non-maximally correlated then two sequences $((t,x))$ and $((t',x'))$ of length $p_1$ and $p_2$ respectively such that $t_{p_1}\neq t'_{p_2}$, can be mapped to the same sequence $(\tau)$ iff $s=s'$ (with $s$ and $s'$ the last terms of $(s)$ and $(s')$). In this case $t_{p_1} = t'_{p_2} + \sum_{k=1}^{P+1}g_{s,k}$ (assuming $t_{p_1}>t'_{p_2}$). \\
\end{coro}

\begin{proof}
	This result derives from proposition 1 and 2.\\
\end{proof}

\begin{coro}
Consider that the functions in the family ${O_x,~ x \in X}$ do not have a fixed point $s_x$ such that $O_x(s_x)=s_x$ and the family of sequences ($g_s$) are non-maximally correlated.\\
Let's denote ($\tau'$) the time code corresponding to the sequence $((t',x'))$ constructed by deleting the last term of $((t,x))$ and ($\tau$) the time code corresponding to $((t,x))$.
We show that $(\tau) \neq (\tau')$.\\
\end{coro}
\begin{proof}
Let's consider the last 2 elements $s_{p-1}$ and $s_p$ of the sequence $(s)$ corresponding to $((t,x))$ with its last element ($t_p,x_p$).\\ 
By construction $s_p=O_{x_p}(s_{p-1})$ and because $O_{x_p}$ doesn't have a fixed point we find that $s_{p-1} \neq s_p$. Corrolary 1 then tells us that in this circumstance $(\tau) \neq (\tau')$.\\
\end{proof}

\begin{prop}\label{prop_tampering}
Consider that the family of sequences ($g_s$) are non-maximally correlated.\\ 
Let's denote ($\tau'$) the time code corresponding to the sequence $((t',x'))$ constructed by changing the last term of $((t,x))$ from ($t_p,x_p$) to ($t_p,x_p'$) and ($\tau$) the time sequence corresponding to $((t,x))$.
$(\tau) = (\tau')$ implies $O_{x_p}(s_{p-1})=O_{x_p'}(s_{p-1})$.\\
\end{prop}

\begin{proof}
 We can easily show the negation of that implication ( $(\tau) = (\tau')$ and $O_{x_p}(s_{p-1}) \neq O_{x_p'}(s_{p-1})$ ) is false.\\
$O_{x_p}(s_{p-1})\neq O_{x_p'}(s_{p-1})$ is $s_p \neq s_p'$ and means that $(g_{s_p}) \neq (g_{s_p'})$ and due to the property of the family of sequence ($g_s$) we have that $(\tau) \neq (\tau')$.\\
\end{proof}

\section{Detection probability of Eve}\label{append2}
Characterisation of the detection probability in the case of an infinite transmission sequence.\\
Consider an infinite sequence ((t,x)). Let's assume that the $d^{th}$ element has been changed from $(t_d,x_d)$ to $(t_d,x_d')$ or $(t_d,\emptyset)$ (deleted). We want to characterise the probability of this event not being detected by the receiver. We'll focus on the case where $x_d$ is changed to $x_d'$ but the case of deletion could be treated similarly adjusting indices accordingly.\\

We have shown previously that if we truncate the original and the modified sequences after the $d^{th}$ element, then they would both be mapped onto different time codes $(\tau)\neq (\tau')$. \\
Even though  $(\tau)\neq (\tau')$ it is possible that $\exists M$ such that $\tau_i=\tau_i', ~ \forall i\leq M$ and $\tau_M > t_d$. Meaning that both time codes still match beyond the time when the message was tampered with. If a new message is received from Alice before $\tau_M$ then it is possible that the newly received message could erase the influence of Eve on the time code. In the following we look at the probability of that happening.\\

Given the structure of the encoding, the change made by Eve would not be detected by the time $t_{d+1}$ when $x_{d+1}$ is received if 
\begin{equation}
 \forall k \mbox{ s.t} \sum_{i=0}^k t_d + g_{s_d,i} < t_{d+1} \mbox{; } g_{s_d,i}=g_{s_d',i} ~ \mbox{ for } i=0 \mbox{ to } k \label{cond1}
\end{equation}
This means that the corresponding elements of $(\tau)$ and $(\tau')$ with values between $t_d$ and $t_{d+1}$ are equal.\\
Similarly, the change isn't detected by the time  $t_{d+2}$ if it wasn't detected beforehand and 
\begin{align}
	&O_{x_{d+1}}(s_d)=O_{x_{d+1}}(s_d'), \mbox{ or} \label{adapt1}\\
	&\forall k \mbox{ s.t } \sum_{i=0}^k t_{d+1} + g_{s_{d+1},i} < t_{d+2} \mbox{; } g_{s_{d+1},i}=g_{s_{d+1}',i} \mbox{ for } i=0 \mbox{ to } k \label{cond2}
\end{align}

The events described by the statements \ref{cond1}, \ref{adapt1} and \ref{cond2} are probabilistic by nature as they depend not only on the coding parameters but also on the information received from the information source. We define the list of probabilities ($P_0,P_1,...$) that statements of type  (\ref{cond1}, \ref{cond2},...) are true and ($Q_1, Q_2...$) that statements of type  \ref{adapt1} are true. 

We can build the following probability tree to determine the probability of Eve avoiding detection.\\
\begin{tikzpicture}[grow=right, sloped,thick,scale=0.7, every node/.style={transform shape}]
\node[bag]{$(t_d,x_d)$}
	child{
		node[bag]{not detected by $t=t_{d+1}$ }
		child{
			node[end, label=right:{undetectable after $t_{d+1}$}]{}
			edge from parent
			node[above]{$Q_1$}
			}
		child{
			node[bag]{not detected by $t_{d+2}$}
			child{
				node[end, label=right:{undetectable after $t_{d+2}$}]{}
				edge from parent
				node[above]{$Q_2$}
			}
			child{
				node[end, label=right:{not detected by $t_{d+3}$...}]{}
				edge from parent
				node[above]{$P_2$}
			}
		edge from parent
		node[above]{$P_1$}
		}
	edge from parent
	node[above]{$P_0$}
		};	
\end{tikzpicture}\\

Note that if the functions $O_x$ are invertible, then $Q_i=0$, $\forall i$. In this case if $P_i<1$ then the probability of Eve avoiding detection will become vanishingly small with time.

\bibliographystyle{IEEEtranN}
\bibliography{biblio}
\end{document}